\documentclass[twoside,11pt]{article}

\usepackage{jmlr2e}
 
\usepackage{subfigure}   
\makeindex

\usepackage{graphicx} 
\usepackage{amsmath}
\usepackage{amssymb}

\usepackage{exscale}

\usepackage{comment}
\usepackage{paralist}

\usepackage{amsmath}%
\usepackage{MnSymbol}%
\usepackage{wasysym}%

\usepackage{algpseudocode}
\usepackage{algorithm}

\usepackage{enumitem}

\usepackage{enumitem}

 \makeatletter
  \def\nl#1#2{\begingroup
     \scalebox{0.85}[1]{\textbf{#2}}%
    \def\@currentlabel{\textnormal{\scalebox{0.85}[1]{\textbf{#2}}}}
     \phantomsection\label{#1}\endgroup
}
\makeatother

\usepackage{enumitem}
\usepackage{geometry}
\usepackage{tikz}
\usetikzlibrary{calc,patterns,decorations.pathmorphing,decorations.markings,arrows}
\usetikzlibrary{quotes,arrows.meta}
\usetikzlibrary{fit}

\tikzstyle{block} = [draw, rectangle, minimum height=3em, minimum width=3em]
\tikzstyle{sum} = [draw, circle, node distance=1.5cm]
\tikzstyle{input} = [coordinate]
\tikzstyle{output} = [coordinate]

\newcommand{\skor}{\mathcal{S}}

\newcommand{\ones}{\mathbf{1}}

\newcommand{\E}{\mathbb{E}}
\newcommand{\obs}{Y}

\newcommand{\reals}{\mathbb{R}}

\renewcommand{\leq}{\leqslant}
\renewcommand{\geq}{\geqslant}
\renewcommand{\epsilon}{\varepsilon}
\renewcommand{\phi}{\varphi}
\renewcommand{\th}{\theta}
\newcommand{\noise}{z}
\newcommand{\wloss}{C}

\newcommand{\loss}{\ell}
\newcommand{\Xt}{X^\th}
\newcommand{\Loss}{L}
\newcommand{\kernel}{K}
\newcommand{\Th}{\Theta}
\newcommand{\CF}{\mathcal{F}}
\newcommand{\1}{\boldsymbol{1}}

\newcommand{\PR}{\mathbb{P}}
\newcommand{\R}{\reals}

\begin{document}

\title{Efficient Counterfactual Estimation of Conditional Greeks via Malliavin-based Weak Derivatives\thanks{This research was supported by NSF grants CCF-2312198 and CCF-2112457 and U. S. Army Research Office under grant W911NF-24-1-0083}}

\author{\name Vikram Krishnamurthy \email vikramk@cornell.edu \\
       \addr School of Electrical and Computer Engineering\\
       Cornell University\\
       Ithaca, NY, USA \AND
       \name Luke Snow \email las474@cornell.edu \\
       \addr School of Electrical and Computer Engineering\\
       Cornell University\\
       Ithaca, NY, USA}

\editor{}

\maketitle



\begin{abstract}
  We study counterfactual gradient estimation of  conditional loss functionals of diffusion processes. In quantitative finance, these gradients are known as conditional Greeks: the sensitivity of expected market values, conditioned on some event of interest. The difficulty is that when the conditioning event has vanishing or zero probability, naïve Monte Carlo estimators are prohibitively inefficient; kernel smoothing, though common, suffers from slow convergence. We propose  a two-stage kernel-free methodology. First, we show using Malliavin calculus that the conditional loss functional of a diffusion process admits an exact
  representation  as a Skorohod integral, yielding classical Monte-Carlo estimator variance and convergence rates.
  Second, we establish that a weak derivative estimate of the conditional loss functional with respect to
  model parameters can be evaluated with constant variance, in contrast to the widely used score function method whose variance grows linearly in the sample path length. Together, these results yield an efficient framework for counterfactual conditional stochastic gradient algorithms and financial Greek computations in rare-event regimes.
\end{abstract}


\section{Introduction and Problem Formulation}

Consider the  stochastic differential equation (SDE)
\begin{equation}
  \label{eq:sde}
  dX_t^{\theta} = b_\th(X_t^{\theta},t) dt + \sigma_{\theta}(X_t^{\theta},t) dW_t, \quad t \in [0,T]
\end{equation}
where $X_t^{\theta} \in \reals^n$, $\theta\in\reals^p$ is a parameter, and $W_t \in \reals^d$ is $d$-dimensional standard Brownian motion.
Denote $X^{\theta} = \{X_t^{\theta}\}_{t\in[0,T]} \in \mathcal{C}[0,T]$, where $\mathcal{C}[0,T]$ is the space of continuous $\reals^n$-valued paths on $[0,T]$. Our aim is to estimate the $\theta$-gradient of the conditional loss function $L(\theta)$:
\begin{equation}
  \label{eq:objective}
  \nabla_{\theta}\Loss(\th) =  \nabla_{\theta}\E[  \ell(X^\th)   \mid  g(X^\th) = 0]
\end{equation}
where $\Th$ is a compact subset of $\reals^p$, $\Loss(\cdot)$ is differentiable w.r.t. $\theta$, and $\loss(\cdot), g(\cdot): \mathcal{C}[0,T] \to \reals$ are functionals. In quantitative finance, the gradient \eqref{eq:objective} is a \textit{conditional} Greek: the sensitivity of an expected market value, conditioned on some event of interest. 

\emph{Counterfactual Estimation.} Counterfactual estimation and learning is widely studied in machine learning -- but typically for multivariate random variables. For example,  given parameter vectors $\alpha$ and $\beta$, counterfactual risk evaluation seeks to evaluate $\E_{p(x|\beta)}\{ \loss(X)\}$ given simulations of $\loss(x)$ where $X$ is drawn from $p(\cdot|\alpha)$. 
Then clearly $$\E_{p(x|\beta)}\{ \loss(X) \} = \E_{p(x|\alpha)} \{ \loss(X) p(X|\beta)/p(X|\alpha) \}. $$
In this paper, we extend this setting to continuous time,  where both the loss and conditioning event are  functionals of the trajectory. Then expressions for $p(X|\beta)$ and $p(X|\alpha)$ are not available in closed form.  The classical  importance-sampling identity relies on  absolute continuity of $p(\cdot|\alpha)$ and $p(\cdot|\beta)$. 
In our continuous-time framework, however, the conditioning event is a zero-probability path
functional ($g(\Xt)=0$), 
so the ratio $p(X|\beta)/p(X|\alpha)$ is no longer meaningful. 
To address this, we exploit the known dynamics of the SDE~\eqref{eq:sde} and  replace the likelihood ratio by a Malliavin calculus representation involving Dirac delta functionals and 
Skorohod integrals, which yields a constructive estimator of the conditional expectation.

The term 'counterfactual' denotes evaluation of the conditional gradient \eqref{eq:objective} using simulated paths which may not necessarily satisfy the conditioning event $\{g(\Xt) = 0\}$ \citep{zenati2025counterfactual,swaminathan2015counterfactual, lopez2020cost}. The usefulness of our proposed approach lies in its ability to evaluate \eqref{eq:objective} under such \textit{counterfactual} conditioning, that is, without necessarily observing paths constrained to the conditioning event. This is accomplished by employing the tools of Malliavin calculus to re-write \eqref{eq:objective} as the ratio of two \textit{unconditional} expectations which can each be evaluated by paths not constrained to the event $\{g(\Xt) = 0\}$. This reformulated evaluation structure provides a significant increase in efficiency by restoring classical Monte Carlo convergence rates, not relying on kernel based methods or  conditional simulation of rare events. 

\textit{Financial applications.} This conditional (counterfactual) evaluation framework is useful for:
\begin{enumerate}[label=\roman*)]
    \item \textit{Barrier/digital and other rare-event exotics:} pricing and computing Greeks conditional on hitting (or being near) a barrier or trigger \citep{Glasserman2003}.
    \item \textit{xVA exposure and stress testing:} estimating expected exposure measures (e.g., CVA) conditional on stressed market scenarios \citep{GregoryXVAChallenge}.
    \item \textit{Tail risk under stress:} evaluating tail-risk metrics (e.g., expected shortfall) conditional on drawdown or high-volatility regimes \citep{AcerbiTasche2002}.
\end{enumerate}

We now give a concrete example: counterfactual evaluation when the asset is forced to be at a stressed level at mid-horizon. This setup leads naturally to \textit{conditional Greeks}, obtained by differentiating the resulting conditional value with respect to model parameters.

\subsection{Example. Conditional Greek Estimation}

Let $X^\theta=\{X_t^\theta\}_{t\in[0,T]}$ \eqref{eq:sde} denote the (risk-neutral) price process of a tradable asset,
parameterized by $\theta$ (e.g., volatility).

We define a path-dependent cumulative loss by
\begin{equation*}
\label{eq:running_loss_example}
\ell(X^\theta)\;:=\;\int_0^T h(X_s^\theta)\,ds,
\end{equation*}
where $h:\R_+\to\R$ is a user-chosen measurable function encoding a financial objective.
Representative choices include:
(i) downside exposure proxy $h(u)=(K-u)^+ := \max\{0,K-u\}$,
(ii) an inventory/hedging penalty proxy $h(u)=\lambda\,\Delta(u)^2$ for $\lambda>0$ and a prescribed position map $\Delta(\cdot)$,
or (iii) a simple excursion penalty $h(u)=u^2$ (or $h(u)=(\log u)^2$).

Next we impose a \emph{counterfactual stress condition} at the intermediate time $t=T/2$ by defining
\begin{equation}
\label{eq:g_mid_example}
g(X^\theta)\;:=\;\log X_{T/2}^\theta-\log s, \qquad s\in\R_+.
\end{equation}
Thus $g(S^\theta)=0$ corresponds to the event $\{X_{T/2}^\theta=s\}$, i.e., the asset price is \emph{exactly} at a stressed level $s$
halfway through the horizon. This is an \emph{anticipatory} constraint, and under mild regularity (e.g., $X_{T/2}^\theta$ admits a density), it is a \emph{zero-probability} event.

We consider the counterfactual conditional risk functional
\begin{equation*}
\label{eq:Loss_finance_example}
\Loss(\theta)\;:=\;\E\!\left[\ell(X^\theta)\,\big|\, g(X^\theta)=0\right]
\;=\;
\E\!\left[\int_0^T h(X_s^\theta)\,ds\ \Big|\ X_{T/2}^\theta=s\right],
\end{equation*}
and the associated counterfactual gradient estimation \eqref{eq:objective}.

More generally, one may wish to impose \emph{path-functional constraints} of the form
\begin{equation*}
\label{eq:g_functional_example}
g(X^\theta)\;=\;\int_0^T \gamma(X_s^\theta)\,ds \;-\; c,
\end{equation*}
for a measurable $\gamma:\R_+\to\R$ and target level $c\in\R$.
This subsumes common finance quantities, e.g., an Asian-average constraint $\gamma(u)=u$,
a time-in-drawdown constraint $\gamma(u)=\mathbf 1\{u\le K\}$, or constraints tied to integrated variance.

\subsection{Computational Intractability}

The computation of the conditional gradient \eqref{eq:objective} is hindered by two main systematic difficulties: 
\begin{enumerate}[label=\roman*)]
    \item there is generally no practical simulation or observational strategy that can \emph{enforce} $g(S^\theta)=0$ exactly
    \item methods for computing \eqref{eq:objective} without such exact conditioning are prohibitively inefficient for rare or measure-zero conditioning events. 
\end{enumerate}

To see this, we can re-express the loss for purposes of stochastic simulation as
\begin{equation}
  \label{eq:equiv_loss}
  \Loss(\th) = \frac{\E\{ \ell(X^\th) \,\delta(g(X^\th)) \}}{\E\{\delta(g(X^\th))\}   }
\end{equation}
where $\delta(g(X))$ denotes the Dirac delta centered at zero. Thus, for $L(\theta)$ to be computable by Monte Carlo, $\{g(X^{\theta})=0\}$ must have non-zero probability. This is certainly not satisfied in many applications of interest (e.g., \eqref{eq:g_mid_example}). Furthermore, even when $\PR(g(X^{\theta})=0) > 0$, methods exploiting classical stochastic simulation have complexity which is inversely proportional to $\PR(g(X^{\theta})=0)$, and thus become practically prohibitive in rare-event regimes. 

This motivates estimators and algorithms that handle conditioning on these counterfactual (often measure-zero) events without requiring constrained path simulation, and which do not suffer from prohibitive complexity in rare- or measure-zero-event regimes.


\subsection{Limitation of Kernel Methods}
Naive Monte-Carlo  estimation of  $\Loss(\th)$ in~\eqref{eq:equiv_loss} fails due to Dirac delta in  the denominator. 
The classical workaround is to use a kernel method: approximate  the Dirac delta  $\delta(g(X))$  by a kernel $\kernel_\Delta(g(X))$ where $\Delta$ denotes the kernel bandwidth. Typically  $\kernel_\Delta$ is  a multivariate Gaussian density and $\Delta$ controls its variance.
The kernel-based Monte-Carlo estimator for the loss $\Loss$ given $N$ independent realizations $X^{(i)},i=1,\ldots,N$ of $X$ is
$$ \hat{\Loss}(\th) = \frac{ \sum_{i=1}^N \ell(\th,X_{[0,T]}^{(i)}) \, \kernel_\Delta(g(X^{(i)})) }{\sum_{i=1}^N  \kernel_\Delta(g(X^{(i)}) )}. $$
But the variance of the estimate of $\Loss(\th)$ depends on the kernel bandwidth $\Delta$ and convergence becomes  excruciatingly slow for large~$n$ or small-probability events $\{g(\Xt)=0\}
$.

\subsection{Main Results} This paper  develops  a two-stage kernel free approach for counterfactual loss functional gradient (Greek) evaluation: \\ 
(i) \textit{Loss evaluation via Malliavin calculus}. We show that  
 $\loss(\Xt) \delta(g(\Xt))$ and $\delta(g(\Xt))$ in~\eqref{eq:equiv_loss} admit exact Skorohod integral representations. Their  expectation can therefore be computed using classical Monte-Carlo.
For $N$ independent trajectories generated by~\eqref{eq:sde}, the estimator
achieves uniform $O(1/N)$ variance, identical to  classical Monte-Carlo, \textit{even in rare- or measure-zero-event settings} \citep{FLL01,BET04,Cri10}.
\\
 (ii) \textit{Gradient estimation via weak derivatives}. We show that the gradient $\nabla_\th \Loss(\th)$ can be estimated efficiently using a weak derivative approach \citep{Pfl96,HV08,KV12} based on the Hahn-Jordan decomposition. The variance of the gradient estimate is $O(1)$. This is in comparison to the widely used score function estimator \citep{kleijnen1996optimization} which has variance $O(T)$.

 \subsection{Context. Counterfactual Stochastic Optimization}

By combining (i) and (ii), we obtain a counterfactual stochastic gradient algorithm that converges to a local stationary point of~$\Loss(\th)$.
The procedure is displayed in Figure~\ref{fig:stochgrad}.

\begin{figure}[h]
  \centering
  \begin{tikzpicture}[scale=0.7, transform shape,node distance = 4.5cm, auto]
    \tikzstyle{arrow} = [thick,->,>=stealth]
    \tikzset{
    block/.style={rectangle, draw, line width=0.5mm, black, text width=4em, text centered,
                 minimum height=2em},
               line/.style={draw, -latex}}
   \tikzset{
    block2/.style={rectangle, draw, line width=0.5mm, black, text width=6em, text centered,
                 minimum height=2em},
               line/.style={draw, -latex}}             

             \node[block](Markov){SDE};
  \node[block2,right of = Markov,node distance=4.2cm](sensor1){Malliavin for loss evaluation};             
  \node[block2,right of = sensor1,node distance=4.2cm](sensor){Weak Derivative};
  \node[block2,right of=sensor, node distance=5.2cm](filter){Stochastic Gradient Algorithm};
  \node[right of=filter,node distance=3.2cm](nullnode){};
  \draw[-Latex](Markov) -- node[above] {$X^{\th_n}$}  (sensor1);
  \draw[-Latex](sensor) -- node[above] {$\hat\nabla_\th \ell(X^{\th_n})$} (filter);
 \draw[-Latex](sensor1) --  (sensor);
  
  \draw[-Latex](filter) -- node[above,pos=0.55] {$\th_{n+1}$} (nullnode);
  \node[draw,inner sep=4pt,dashed,fit={(sensor1) (filter)},label={Counterfactual Stochastic Approximation   Algorithm}] {};
  \draw[-Latex] (filter.east) -- ++(1,0) |- ([yshift=-0.5cm]Markov.south)
  -|
  (Markov.south);
\end{tikzpicture}
\caption{Counterfactual Stochastic Gradient Algorithm}
\label{fig:stochgrad}
\end{figure}
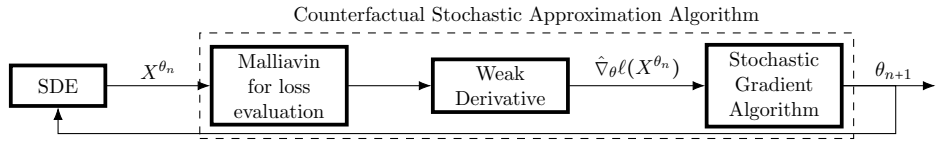

{\em Remark: Model fitting}.
 The above gradient estimation framework can be used in a stochastic optimization procedure which
 aims to choose $\th$  to control  SDE~\eqref{eq:sde} to minimize the conditional loss $\Loss(\th)$. The framework also applies to model fitting: fit 
 the SDE~\eqref{eq:sde} to $N$ externally  generated data trajectories  $\obs_{0:T}^{(1)},\ldots \obs_{0:T}^{(N)}$. In this case, one seeks to minimize the conditional loss  $L(\th) = \E\{\loss(Y,\Xt) | g(\Xt) = 0 \}$.  

 
 In classical stochastic approximation we  observe a sequence of noisy gradients $\{\nabla \loss(\th_k,\noise_k)\}$ where $\noise_k$ is a noisy signal and  $\nabla \loss(\th,\noise) $ is an asymptotically unbiased estimate of $\nabla \Loss(\th)$. We optimize
$\Loss(\th) = \E\{\ell(\th_k,\noise_k)\}$ via the stochastic gradient algorithm  
$$ \th_{k+1} = \th_k - \epsilon \nabla \ell(\th_k,\noise_k) $$
Under reasonable conditions \citep{KY03}, the interpolated trajectory of the estimate $\{\th_k\}$ converges weakly to the ordinary differential equation (ODE)
\begin{equation}
  \label{eq:ODE}
 \frac{d\th}{dt} = -\nabla L(\th) . 
\end{equation}    

\subsubsection{Passive Learning}
    In passive stochastic optimization \citep{NPT89,YY96,KY22,snow2025finite}, unlike classical stochastic gradient, we observe a sequence of noisy and {\em misspecified gradients}: $\{\alpha_k, \nabla \loss(\alpha_k,\noise_k)\}$, where the parameters $\alpha_k\in \reals^p$ are chosen randomly   according to  probability density $p(\cdot)$, potentially by an adversary.  The passive stochastic gradient algorithm is 
    \begin{equation}
      \label{eq:passive}
    \th_{k+1} = \th_k - \epsilon\,\kernel_\Delta(\alpha_k-\th_k )\,\nabla \ell(\alpha_k,\noise_k) . 
  \end{equation}
  The kernel $\kernel_{\Delta}(\cdot)$ weights the usefulness of the gradient $\nabla\loss(\alpha_k,\noise_k)$ compared to the required gradient
$\nabla\loss(\th_k,\noise_k)$.
If $\th_k$ and $\alpha_k$ are far apart, then kernel is   small and  only a small proportion of the gradient estimate $\nabla\loss(\alpha_k,\noise_k)$ is added to the stochastic gradient algorithm. On the other hand, if $
\alpha_k = \th_k$,
the algorithm  becomes the  classical  stochastic gradient algorithm. 
Under reasonable conditions, for small bandwidth parameter $\Delta$, the kernel $\kernel_\Delta$ behaves as a Dirac delta and   the interpolated trajectory converges weakly to the ODE
$$\frac{d \th}{dt} = -\int_\Theta \pi(\theta) \,\delta(\alpha- \th)\, \nabla L(\alpha)\, d\alpha  =-p(\th) \, \nabla L(\th).
$$
Notice that this ODE has the same fixed points as~\eqref{eq:ODE}.

\subsubsection{Counterfactual Learning}
Finally, the counterfactual stochastic optimization problem described above, can be regarded as a passive stochastic optimization problem.  At each stage, we require gradient estimates  
$\nabla_\th \loss(\th,g(\Xt)=0,\noise) $ that are unbiased estimates of
$\nabla_\th \Loss(\th) $ where $\Loss(\th)$ is defined in~\eqref{eq:objective}, but we are instead provided with  noisy and misspecified gradient estimates 
$\nabla_\th \loss(\th,g(\Xt)=a,\noise) $ for random  $a \in \reals$. That is, while
the desired gradient corresponds to the counterfactual constraint
  $[\th,g(\Xt)=0]$,  we only observe  the misspecified noisy gradient evaluated at
$\alpha = [\th,g(\Xt)=a]$.  Therefore, one can use the passive kernel based algorithm~\eqref{eq:passive} to solve the counterfactual stochastic optimization problem.  However, in this paper, we will exploit the structure of the SDE~\eqref{eq:sde} and not use the kernel based algorithm.

\section{Malliavin Calculus Approach to Estimate Conditional Loss}

Malliavin calculus \citep{Nua06} was  developed in the 1970s as a probabilistic method to prove H\"ormander hypoellipticity theorem for the solution of SDEs.  It was later adapted in mathematical finance to compute sensitivities (Greeks) of option prices.  Here, as in \citep{FLL01,BET04,Cri10}, we employ  Malliavin calculus  to \textit{efficiently} evaluate the conditional loss $\E\{\loss(\Xt)|g(\Xt)=0\}$, even when the conditioning event $\{g(\Xt)=0\}$ has vanishing or zero probability. 


\subsection{Preliminaries. Malliavin Calculus.}
\label{sec:mall_pre}

We briefly recall the two central objects.

\subsubsection{Malliavin derivative}
 We work on the probability space $(\Omega,\mathcal{F},\mathbb{P})$ with  a $d$-dimensional Brownian motion $W = (W^1,\dots,W^d)$ and the natural filtration $\{\CF_t\}_{t\geq0}$. 
For a smooth functional $F$ of $W$, the \emph{Malliavin derivative} $D_t F$ is defined as the process measuring 
the infinitesimal sensitivity of $F$ to perturbations of the Brownian path at time $t$. Formally, 
for cylindrical random variables of the form
\[
F = f\bigg( \int_0^T h_1(s)\, dW_s, \dots, \int_0^T h_n(s)\, dW_s \bigg),
\]
with $f \in C_b^\infty(\mathbb{R}^n)$ and $h_i \in L^2([0,T];\mathbb{R}^d)$, the derivative is
\[
D_t F = \sum_{i=1}^n \frac{\partial f}{\partial x_i}\bigg( \int_0^T h_1\, dW, \dots, \int_0^T h_n\, dW \bigg)\, h_i(t).
\]
The closure of this operator in $L^p$ leads to the Sobolev space $\mathbb{D}^{1,p}$ of Malliavin differentiable random variables.

\subsubsection{Skorohod integral}
The adjoint of the Malliavin derivative is the \emph{Skorohod integral}, denoted $\skor(u)$.  Indeed, for a process $u \in L^2([0,T]\times \Omega;\mathbb{R}^d)$, $u$ is in the domain of $\skor$ if there exists a square-integrable 
random variable $\skor(u)$ such that for all $F \in \mathbb{D}^{1,2}$,
\begin{equation}
\label{eq:adjoint}
\E[F \, \skor(u)] = \E\!\left[\int_0^T \langle D_t F, u_t \rangle \, dt \right].
\end{equation}
The above adjoint relationship  serves as the definition of the Skorohod integral and can be written abstractly as  $$\langle F, \mathcal{S}(u) \rangle_{L^2(\Omega)} =
\langle D F, u \rangle_{L^2([0,T] \times \Omega)}.$$

When $u$ is adapted to the filtration $\{\CF_t\}_{t\geq0}$, the Skorohod integral $\skor(u)$ coincides with the Itô integral $\int_0^T u_t\, dW_t$. 
In general, $\skor(u)$ extends stochastic integration to non-adapted processes and is sometimes called the \emph{divergence operator}.

\subsubsection{Integration by parts}
The duality relation~\eqref{eq:adjoint}  yields the Malliavin’s integration-by-parts formula, which underpins many applications, 
including Monte Carlo estimation of conditional expectations and sensitivity analysis for SDEs 
(see \citep{Nua06,FLL01,BET04}). 

\subsubsection{Computing Malliavin Derivative and Skorohod Integral} 
\label{sec:mall_comp}
The following properties are the key tools which allow us to compute the Malliavin derivative and Skorohod integral: 
\begin{enumerate}
    \item \textit{Malliavin derivative of diffusion}. For  diffusion process $\{X_t\}_{t\geq 0}$ \eqref{eq:sde}, the  Malliavin derivative $D_sX_t$ is \citep{gobet2005sensitivity}
    \begin{equation}
    \label{eq:mall_form}
        D_sX_t = Y_tZ_s\sigma(X_s,s)\1_{s\leq t}
    \end{equation}
    where  $Y_t:= \nabla_x X_t$ is the Jacobian matrix and $Z_t$ is its inverse $Z_t := Y_t^{-1}$.  This, together  with the Malliavin chain rule, facilitates evaluating  Malliavin derivatives of general functions of diffusions.
    \item \textit{Skorohod expansion}. For random variable $F\in \mathbb{D}^{1,2}$ and Skorohod-integrable process $u$, we have \citep[eq. 2.2]{gobet2005sensitivity}:
    \begin{equation}
    \label{eq:skor_exp}
    \skor(F u) = F\skor(u) - \int_0^TD_tF\cdot u_t dt
    \end{equation}
    In general,  the Skorohod integrand $\{u_t\}_{t\in[0,T]}$ of interest is non-adapted. However, in the special case where $u$ factorizes into  the product of an adapted process $\hat{u}=\{\hat{u}_t\}_{t\in[0,T]}$ and an anticipative random variable $F$, this formula gives a constructive expression.  Specifically, we can expand $\skor(u) = \skor(F\hat{u}) $ using  \eqref{eq:skor_exp} and compute it in terms of a  standard It\^o integral of the adapted part $\hat{u}$ together with the  Malliavin derivative of the anticipatory random variable $F$. 
    
\end{enumerate}

\subsection{Malliavin Calculus Expression for Conditional Expectation}
The following main result expresses   the conditional expectation \eqref{eq:objective} as the ratio of unconditional expectations. 
\begin{theorem}  Assume
$\ell(\Xt), g(\Xt) \in L^2(\Omega)$ and  $D_t\ell(\Xt), D_tg(\Xt) \in L^2(\Omega \times [0,T])$. Then  the conditional loss $\Loss$ in~\eqref{eq:equiv_loss} is
\begin{align}
\begin{split}
\label{eq:malliavin}
 &\Loss(\th)  =\E[ \ell(X^\th)  \mid g(X^\th) = 0]  =  \frac{E_1^\th}{E_2^\th} \\
   &\text{ where } \\&E_1^\th 
 = \E\bigg[ \ones_{\{g(\Xt)>0\}}\left(\ell(X^\th) \skor(u)  - 
  \int_0^T (D_t \ell(\Xt)) u_t dt\right) \bigg]\\  &E_2^\th =
\E[  \ones_{\{g(\Xt)>0\}} \skor(u)]
\end{split}
\end{align}
Here 
$u$ is any process that satisfies  
\begin{equation}
\label{eq:ut_cond}
\E[\int_0^T D_tg(\Xt) u_t\,dt] = 1
\end{equation}
\end{theorem}

\textbf{Proof outline}: We start with \eqref{eq:equiv_loss} and write $\delta(g(\Xt))$ as $\delta(G)$.
Then, by the Malliavin chain rule, the adjoint relation \eqref{eq:adjoint} and the Skorohod integrand condition \eqref{eq:ut_cond}, we have
\begin{align*}
&\E[\ell(\Xt)\,\delta(G)] \\&
=\E\!\left[\int_0^T (D_t(\ell(\Xt))\1_{\{g(\Xt)>0\}}))\,u_t\,dt\right]
\\&= \E\!\Big[\1_{\{g(\Xt)>0\}}\left(\ell(\Xt)\skor(u)-
\int_0^T (D_t\ell(\Xt))u_t\,dt\right)\Big].
\end{align*}
The denominator in \eqref{eq:equiv_loss} can be derived similarly. 

\paragraph*{Remarks.}
(i) There is considerable flexibility  in the choice of $u$ in the above theorem.
For example if $D_t g(\Xt) \neq 0$ a.e., one can choose
\begin{equation}
\label{eq:u_choice}
u_t = \begin{cases} \frac{1}{T D_t g(\Xt)}  & D_t g(\Xt) \neq 0 \\
     1 &  D_t g(\Xt) = 0 .
     \end{cases}
\end{equation}

(ii) 
The 
representation~\eqref{eq:malliavin} requires evaluation of Malliavin derivatives and Skorohod integrals, see \citep{gobet2005sensitivity} for several examples. There are  several important consequences.  First, it restores the
$N^{-1/2}$ Monte--Carlo convergence rate even under singular conditioning, as
the event $\{g(X^\theta)=0\}$ no longer needs to be sampled directly.  Second,
the estimator admits substantial variance--reduction flexibility: the choice of
localizing function (indicator versus smooth approximation) and of admissible
weight process $u$ strongly influence efficiency, with optimal choices
characterizable via variational principles in Malliavin calculus.  Third, the
representation is compatible with standard discretizations of the forward SDE:
the Malliavin derivatives $D_t X^\theta$ admit recursive Euler--Maruyama
approximations, so one avoids additional kernel bandwidths or curse--of--dimensionality
issues inherent in regression--based methods.

\section{Weak Derivative Estimator}

Applying the quotient rule, 
it follows  from~\eqref{eq:malliavin} that
\begin{equation}
\label{eq:mall_grad}
\nabla_\theta \E[\ell(X^\theta)\mid g(\Xt) = 0]
= \frac{E_2\,\nabla_\theta E_1 - E_1\,\nabla_\theta E_2}{E_2^2}.
\end{equation}
In this section we construct a weak derivative based algorithm to estimate $\nabla_\th E_1$ and
$\nabla_\th E_2$ given the SDE~\eqref{eq:sde}. The resulting gradient estimate can then be fed into a stochastic gradient algorithm to minimize the loss $\Loss(\th)$. This weak-derivative method recovers an $O(1)$ variance scaling with respect to the time horizon $T$, in contrast to score function methods which incur $O(T)$ variance scaling.

\subsection{Discrete-Time Measure-Valued Weak Derivative}
The purpose of this section is to derive a computationally feasible weak-derivative estimator for the gradient \eqref{eq:objective}. We begin with the Euler discretization of the SDE \eqref{eq:sde}, and derive a weak derivative of the time-inhomogeneous Markov kernel governing its dynamics.

For practical computation, we work with an Euler discretization of the SDE~\eqref{eq:sde} on a grid $t_k = k \Delta t$, $k=0,\dots,M$, with $\Delta t = T/M$. Let $\Sigma(x,t) := \sigma_{\theta}(x,t)\sigma_{\theta}(x,t)^\top$. The resulting (time-inhomogeneous) Markov chain
\[
X_{k+1}^\theta = X_k^\theta + \Delta t\,b_\theta(X_k^\theta,t_k) 
               + \sqrt{\Delta t}\,\sigma_{\theta}(X_k^\theta,t_k)\,\xi_{k+1}, 
\qquad \xi_{k+1}\sim\mathcal N(0,I_d),
\]
has one-step transition kernels
\begin{align}
\begin{split}
\label{eq:gauskern}
P_{\Delta t,k}^\theta(x,t_k,dx')
= \mathcal N\big(x+\Delta t\, b_\theta(x,t_k),\,\Delta t\,\Sigma(x,t_k)\big).
\end{split}
\end{align}
For each time index $k$ and parameter $\theta$, $P_{\Delta t,k}^\theta(x,t_k,\cdot)$ is a probability measure on $\reals^n$. Its parameter derivative is not a probability measure in general, but a \emph{finite signed measure}. More precisely, if
\[
\nabla_{\theta} P_{\Delta t,k}^\theta(x,t_k,\cdot)
\]
exists in the weak sense\footnote{To keep the notation simple and avoid multidimensional matrices, we assume $\theta$ is a scalar parameter. Dealing with $\th \in\reals^p$ simply amounts to interpreting the results elementwise.} then for every bounded, smooth test function $f$,
\begin{equation}
\label{eq:wd}
\nabla_{\theta} P_{\Delta t,k}^\theta f(x,t_k) 
\;=\; \nabla_{\theta}\int_{\reals^n} f(x') \, P_{\Delta t,k}^\theta(x,t_k,dx').
\end{equation}
We call~\eqref{eq:wd} the \emph{weak derivative} (or \emph{measure-valued derivative}) of $P_{\Delta t,k}^\theta$ at time $t_k$.

By the Hahn–Jordan decomposition theorem, any finite signed measure $\nu$ on a measurable space can be expressed as the difference of two mutually singular positive measures: $\nu = \nu^+ - \nu^-$, with $\nu^+$ and $\nu^-$ uniquely determined. Applying this to the weak derivative
$\nabla_{\theta} P_{\Delta t,k}^\theta(x,t_k,\cdot)$ at each time step, we obtain a pointwise decomposition
\[
\nabla_{\theta} P_{\Delta t,k}^\theta(x,t_k,dx') 
= c_{\theta,k}(x,t_k)\big(\rho^+_{\theta,k}(x,t_k,dx') - \rho^-_{\theta,k}(x,t_k,dx')\big),
\]
where $\rho^\pm_{\theta,k}(x,t_k,\cdot)$ are probability measures (the normalized Hahn–Jordan components) and $c_{\theta,k}(x,t_k)$ is a finite scalar weight.

We formalize this as follows.

\begin{theorem}[Discrete-Time Hahn--Jordan Weak Derivative]
\label{thm:hj_discrete}
Let $(X_t^\theta)_{t\in[0,T]} \subset \reals^n$ solve the Itô SDE \eqref{eq:sde} with 
$b_\theta \in C_b^2(\reals^n\times \reals;\reals^n)$ and 
$\sigma_{\theta} \in C_b^2(\reals^n \times \reals;\reals^{n\times d})$, 
where $d$ is the dimension of the driving Brownian motion.  

For $\Delta t > 0$, the Euler--Maruyama scheme induces the Gaussian transition probabilities \eqref{eq:gauskern} with time-inhomogeneous kernels $\{P_{\Delta t,k}^\theta\}_{k=0}^{M-1}$. Then, for each $k$ and $(x,t_k)$, the weak derivative of $P^\theta_{\Delta t,k}$ with respect to~$\theta$ admits a Hahn--Jordan decomposition
\begin{equation}
\label{eq:hjd}
\nabla_{\theta} P^\theta_{\Delta t,k}(x,t_k,dx') 
= c_{\theta,k}(x,t_k)\big(\rho^+_{\theta,k}(x,t_k,dx') - \rho^-_{\theta,k}(x,t_k,dx')\big),
\end{equation}
where $\rho^\pm_{\theta,k}(x,t_k,\cdot)$ are mutually singular positive measures and $c_{\theta,k}(x,t_k)\in\reals$ is finite. Consequently, for any bounded measurable $f:\reals^n\to\reals$,
\begin{align*}
&\nabla_{\theta} \int f(x')\,P^\theta_{\Delta t,k}(x,t_k,dx')\\
&\quad = c_{\theta,k}(x,t_k)\left(\int f(x')\,\rho^+_{\theta,k}(x,t_k,dx') 
             - \int f(x')\,\rho^-_{\theta,k}(x,t_k,dx')\right).
\end{align*}
\end{theorem}

\begin{proof}
Since $P^\theta_{\Delta t,k}$ is Gaussian with mean 
$\mu_\theta(x,t_k) = x + \Delta t\,b_\theta(x,t_k)$ and covariance 
$\Sigma(x,t_k) = \Delta t\,a(x,t_k)$, the density $p_\theta(x,t_k,x')$ is smooth in $\theta$ by 
the $C^2_b$ assumption.\footnote{$C_b^2$ denotes twice-differentiable bounded functions, and $C_c^{\infty}$ denotes infinitely-differentiable functions with compact support.} For any $\phi \in C_c^\infty(\reals^n)$,
\[
\nabla_{\theta} \int \phi(x')\,P^\theta_{\Delta t,k}(x,t_k,dx') 
= \int \phi(x')\, \nabla_{\theta} p_\theta(x,t_k,x')\,dx',
\]
so $\nabla_{\theta} P^\theta_{\Delta t,k}(x,t_k,\cdot)$ defines a finite signed measure. 
By the Hahn--Jordan decomposition theorem, every finite signed measure admits 
a unique decomposition into two mutually singular positive measures 
$\nu^+$ and $\nu^-$. Normalizing these gives probability measures $\rho^+_{\theta,k}$ and $\rho^-_{\theta,k}$ and a scalar weight $c_{\theta,k}(x,t_k)$ so that~\eqref{eq:hjd} holds. Finally, since
\[
\int\nabla_{\theta}P_{\Delta t,k}^{\theta}(x,t_k,dx') 
= \nabla_{\theta}\int P_{\Delta t,k}^{\theta}(x,t_k,dx') = 0,
\]
\end{proof}

\begin{figure}
  \centering
  \includegraphics[scale=0.3]{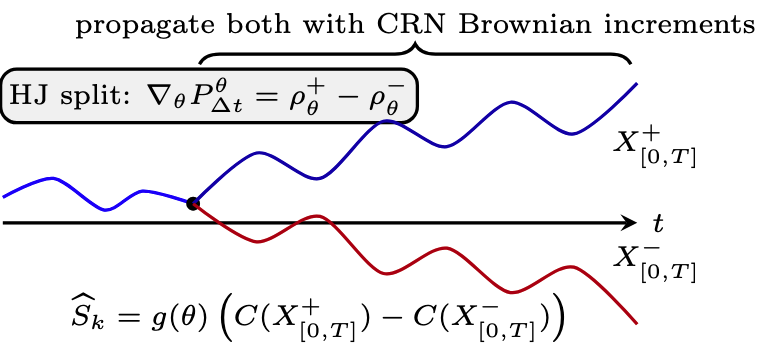}
  \caption{Conceptual schematic of the Hahn--Jordan decomposition for the weak derivative estimator. At each time step, the derivative of the Euler kernel is represented as a signed difference of two branched transitions driven by common random numbers (CRN).}
\end{figure}

\subsection{Finite-Horizon Splitting Representation}

In this section we derive a representation of the weak derivative as a cumulative difference of weighted paths which are split at each time-point \citep{heidergott2000measure}. This provides an implementable algorithm for computing weak derivatives of functionals of \textit{non-stationary} processes, which is necessary for most applications. 

Let $\mu$ denote an arbitrary initial distribution for $X_0^\theta$, and let $\Xt = (X_0^\theta,\dots,X_M^\theta)$ denote the Euler path. For a bounded path-functional $\wloss:\reals^{(M+1)n}\to\reals$ we consider
\[
J(\theta) := \E_\mu[\wloss(\Xt)].
\]
Measure-valued differentiation provides a product rule for the family of kernels $\{P_{\Delta t,k}^\theta\}_{k=0}^{M-1}$: the derivative of the finite product of kernels can be expanded as a sum over all discrete times at which the derivative acts. Concretely, one obtains the representation
\begin{equation}
\label{eq:mvd-product}
\nabla_\theta J(\theta)
= \sum_{k=0}^{M-1} 
\E_\mu\Big[
  c_{\theta,k}(X_k^\theta,t_k)
  \big(\wloss(X^{\theta,+,k}_{0:M}) - \wloss(X^{\theta,-,k}_{0:M})\big)
\Big],
\end{equation}
where, for each $k$, the \emph{phantom chains}\citep{heidergott2000measure} $X^{\theta,\pm,k}_{0:M}$ are defined as follows:
\begin{itemize}
\item[-] For $i \le k$,
\(
X^{\theta,\pm,k}_i = X^\theta_i
\)
(i.e., both phantoms and the nominal chain coincide up to time $t_k$).
\item[-] At time $t_{k+1}$ we branch using the Hahn--Jordan components:
\[
X^{\theta,+,k}_{k+1} \sim \rho^+_{\theta,k}(X_k^\theta,t_k,\cdot), 
\qquad
X^{\theta,-,k}_{k+1} \sim \rho^-_{\theta,k}(X_k^\theta,t_k,\cdot),
\]
conditionally on $X_k^\theta$.
\item[-] For $i\ge k+1$, both phantoms evolve under the nominal kernel $P_{\Delta t,i}^\theta$, using the same Gaussian increments (CRN coupling) as each other.
\end{itemize}
Equation~\eqref{eq:mvd-product} is the discrete-time analogue of splitting the signed derivative at every point in time: at each step $k$, the derivative of the kernel injects a signed mass $c_{\theta,k}(\cdot)$ which is propagated forward through the remaining dynamics.

\subsection{Single-Run Hahn--Jordan Weak-Derivative Estimator}
\label{sec:HJWD}

The representation~\eqref{eq:mvd-product} suggests a direct Monte Carlo implementation that would, in principle, branch at every time step $k$ and average the contributions. This \emph{phantom} construction is conceptually useful but computationally expensive, since the number of branched paths grows with $M$. In practice, we follow the single-run MVD paradigm and realize~\eqref{eq:mvd-product} with only one branch per simulated trajectory.

Let $q$ be a probability mass function on $\{0,\dots,M-1\}$ (for example, $q_k = 1/M$). For each Monte Carlo replication:
\begin{enumerate}
\item Simulate a nominal Euler path $(X_0^\theta,\dots,X_M^\theta)$ starting from $X_0^\theta\sim\mu$.
\item Draw a branch index $K\in\{0,\dots,M-1\}$ with $\mathbb{P}(K=k) = q_k$.
\item At time $t_K$, form the local Hahn--Jordan decomposition 
\(
\nabla_\theta P_{\Delta t,K}^\theta(X_K^\theta,t_K,\cdot)
= c_{\theta,K}(X_K^\theta,t_K)\big(\rho^+_{\theta,K} - \rho^-_{\theta,K}\big)
\)
and sample
\[
X_K^{\theta,(+)}\sim\rho^+_{\theta,K}(X_K^\theta,t_K,\cdot),\qquad
X_K^{\theta,(-)}\sim\rho^-_{\theta,K}(X_K^\theta,t_K,\cdot).
\]
\item Generate future Gaussian increments $\{\xi_j\}_{j=K+1}^M$ and reuse them for both branches:
propagate $X^{\theta,(+)}$ and $X^{\theta,(-)}$ from $t_{K+1}$ to $T$ using the same noise sequence (CRN coupling),
thereby obtaining branched paths $X^{\theta,(+)}_{0:M}$ and $X^{\theta,(-)}_{0:M}$ that coincide with the nominal path up to $t_K$.
\end{enumerate}

\begin{algorithm}[H]\small
\caption{Single–run HJ weak-derivative estimator (Euler kernels, non-stationary initial law $\mu$)}
\label{alg:hjd}
\begin{algorithmic}[1]
\Require horizon $T$, steps $M$, $\Delta t=T/M$, parameter $\theta$, branch law $q$, path functional $\wloss$
\State Sample $X_0^\theta\sim\mu$ and simulate the Euler path $X_0^\theta,\dots,X_M^\theta$.
\State Sample branch index $K\in\{0,\dots,M-1\}$ with $\mathbb{P}(K=k)=q_k$.
\State Form $c_{\theta,K}(X_K^\theta,t_K)$ and $\rho_{\theta,K}^\pm(X_K^\theta,t_K,\cdot)$ from~\eqref{eq:hjd}.
\State Draw $X_K^{\theta,(+)}\sim\rho^+_{\theta,K}(X_K^\theta,t_K,\cdot)$, $X_K^{\theta,(-)}\sim\rho^-_{\theta,K}(X_K^\theta,t_K,\cdot)$.
\State Generate $\{\xi_j\}_{j=K+1}^M$ and reuse them for both branches (CRN), propagating to obtain $X^{\theta,(+)}_{0:M}$ and $X^{\theta,(-)}_{0:M}$.
\State Return the single-run estimator
\begin{equation}
\label{eq:s_k}
    \widehat S 
    \;=\; \frac{c_{\theta,K}(X_K^\theta,t_K)}{q_K}
          \Big(\wloss(X^{\theta,(+)}_{0:M})-\wloss(X^{\theta,(-)}_{0:M})\Big).
\end{equation}
\end{algorithmic}
\end{algorithm}

\textbf{Remarks.} By construction and~\eqref{eq:mvd-product}, $\E_\mu[\widehat S]=\nabla_\theta J(\theta)$ for any choice of $q$ with full support on $\{0,\dots,M-1\}$. Thus Algorithm~\ref{alg:hjd} realizes the \emph{split-at-every-point} measure-valued derivative in a single-branch Monte Carlo estimator. The CRN coupling ensures that the difference
$\wloss(X^{\theta,(+)}_{0:M})-\wloss(X^{\theta,(-)}_{0:M})$ is driven by a localized perturbation at $t_K$ and remains $O(1)$ in $T$ under mild ergodicity assumptions, yielding an $O(1)$ variance scaling in the time horizon.

\subsubsection{Malliavin Gradient Estimation}

We now connect the weak-derivative estimator to the Malliavin-gradient representation~\eqref{eq:mall_grad}. For the counterfactual stochastic gradient algorithm, recall that we need to compute sensitivities $\nabla_{\theta} E_1^{\theta}$ and $\nabla_{\theta} E_2^{\theta}$, defined in \eqref{eq:malliavin}. We compute these by applying Algorithm~\ref{alg:hjd} with path-functionals
\begin{equation*}
\wloss_1(\Xt) 
= \ones_{g(\Xt)>0}\left(\ell(X^\th) \skor(u) -
  \int_0^T (D_t \ell(\Xt)) u_t \, dt \right),
\end{equation*}
and
\begin{equation*}
\wloss_2(\Xt) 
= \ones_{\{g(\Xt)>0\}} \skor(u),
\end{equation*}
respectively, where $u$ solves \eqref{eq:ut_cond}. Notice that in Algorithm~\ref{alg:hjd}, we only need to plug in these loss functionals, and not their derivatives w.r.t.\ $\theta$, to compute $\nabla_{\theta} E_1^{\theta}$ and $\nabla_{\theta} E_2^{\theta}$. Computation of $\skor(u)$ and $D_t\ell(\Xt)$ is attained as described in \eqref{eq:mall_form} and \eqref{eq:skor_exp} in Section~\ref{sec:mall_comp}; see, e.g., \citep{gobet2005sensitivity} for implementation details. In Section~\ref{sec:num} we illustrate such computations for a Black-Scholes model.

\subsection{Connection to the Infinitesimal Generator}

The weak-derivative estimator is traditionally developed for discrete-time Markov chains. The aim of this section is to connect the Euler/HJ construction to the continuous-time diffusion via the infinitesimal generator and its Fokker–Planck equation. Two complementary viewpoints underlie the method:

\subsubsection{Infinitesimal Generator Formulation}

Let $p_t^\theta$ denote the density of $X_t^\theta$ and $L^\theta$ the generator of~\eqref{eq:sde}. The Fokker--Planck equation reads
\[
\partial_t p_t^\theta = (L^\theta)^{\!*}p_t^\theta.
\]
Differentiating with respect to~$\theta$ yields the sensitivity PDE
\[
\partial_t \nu_t = (L^\theta)^{\!*}\nu_t + \nabla_{\theta} (L^\theta)^{\!*}p_t^\theta,
\qquad \nu_t := \nabla_{\theta} p_t^\theta.
\]
By Duhamel’s principle,
\[
\nu_T = \int_0^T P_{T-s}^\theta \big( \nabla_{\theta} (L^\theta)^{\!*}p_s^\theta \big)\,ds,
\]
so the derivative measure at time~$T$ is an integral of \emph{signed mass injections}
$\nabla_{\theta} (L^\theta)^{\!*}p_s^\theta$ transported forward by the semigroup $P^\theta$.
A Hahn--Jordan decomposition can be applied to each signed measure 
$\nabla_{\theta} (L^\theta)^{\!*}p_s^\theta$.

\subsubsection{Discrete Euler Formulation.}
The Euler--Maruyama discretization induces kernels $P^\theta_{\Delta t,k}$ of the form~\eqref{eq:gauskern}. The weak derivatives $\nabla_{\theta} P^\theta_{\Delta t,k}$ are finite signed measures that admit Hahn--Jordan decompositions~\eqref{eq:hjd}. The finite-horizon product rule then yields the split-at-every-point representation~\eqref{eq:mvd-product}, whose single-run realization is Algorithm~\ref{alg:hjd}.

\subsubsection{Consistency.}
As $\Delta t \to 0$, the Riemann-sum representation
\begin{align*}
    \sum_{k=0}^{M-1} 
    P^\theta_{T-t_{k+1}}\big(\nabla_{\theta} P^\theta_{\Delta t,k}\big) P^\theta_{t_k}
\;\longrightarrow\;
\int_0^T P^\theta_{T-s}\big(\nabla_{\theta} (L^\theta)^{\!*}p_s^\theta\big)\,ds
\end{align*}
formalizes the equivalence between the generator-based and Euler-based descriptions. Thus the Euler/HJ scheme is a Monte Carlo realization of the generator-level Hahn–Jordan decomposition, valid for non-stationary initial conditions, with the same $O(1)$ variance properties but amenable to practical implementation.

\section{Numerical Implementation. Conditional Black-Scholes}
\label{sec:num}

Let $(\Omega,\CF,(\CF_t)_{t\in[0,T]},\PR)$ support a 1D Brownian motion $(W_t)_{t\in[0,T]}$.
Under a risk-neutral measure, consider the Black--Scholes model
\begin{equation*}
\label{eq:bs_S}
dS_t^\theta = r S_t^\theta\,dt + \theta S_t^\theta\,dW_t,\qquad S_0>0,
\end{equation*}
and define the log-price $X_t^\theta := \log S_t^\theta$, so that
\begin{equation}
\label{eq:bs_X}
dX_t^\theta = \Big(r-\tfrac12\theta^2\Big)\,dt + \theta\,dW_t,\qquad X_0=\log S_0.
\end{equation}
Fix a stressed intermediate level $s>0$ and define the anticipatory constraint
\begin{equation*}
\label{eq:g_bs}
g(X^\theta) := X_{T/2}^\theta - \log s.
\end{equation*}
As in Theorem~1, the counterfactual conditional loss is
\begin{equation}
\label{eq:L_bs_def}
L(\theta) := \E\!\left[\ell(X^\theta)\mid g(X^\theta)=0\right]
            = \frac{E_1^\theta}{E_2^\theta},
\end{equation}
where
\begin{align}
\begin{split}
\label{eq:E1E2_general}
E_1^\theta &:= \E\!\left[\1_{\{g(X^\theta)>0\}}
\Big(\ell(X^\theta)S(u)-\int_0^T (D_t\ell(X^\theta))\,u_t\,dt\Big)\right],\\
E_2^\theta &:= \E\!\left[\1_{\{g(X^\theta)>0\}}\,S(u)\right],
\end{split}
\end{align}
and $u$ satisfies the normalization condition $\E[\int_0^T (D_t g(X^\theta))u_t\,dt]=1$.

\subsection{Choice of payoff and parameter.}
Let $\ell$ be the discounted European call payoff
\begin{equation*}
\label{eq:call_payoff}
\ell(X^\theta) := e^{-rT}\,(S_T^\theta-K)^+ = e^{-rT}\,(e^{X_T^\theta}-K)^+,
\end{equation*}
and treat $\theta$ as the parameter of interest (so $\nabla_\theta L(\theta)$ is a
conditional \emph{vega} under the stressed event $S_{T/2}=s$).

\subsection{Malliavin derivatives for the log-price model.}
Since \eqref{eq:bs_X} is affine in $W$, its Malliavin derivative is explicit:
\begin{equation*}
\label{eq:DX}
D_t X_\tau^\theta = \theta\,\1_{\{t\le \tau\}},\qquad 0\le t,\tau\le T.
\end{equation*}
Therefore,
\begin{equation*}
\label{eq:Dg}
D_t g(X^\theta)=D_tX_{T/2}^\theta=\theta\,\1_{\{t\le T/2\}}.
\end{equation*}

Because $D_t g$ is deterministic and adapted, we choose an adapted $u$ supported on $[0,T/2]$:
\begin{equation*}
u_t := \frac{2}{\theta T}\,\1_{\{t\le T/2\}}.
\end{equation*}
Then $\int_0^T (D_t g)u_t\,dt=\int_0^{T/2}\theta\cdot \frac{2}{\theta T}\,dt=1$, so the
normalization condition holds.
Since $u$ is adapted, $S(u)$ coincides with an It\^{o} integral:
\begin{equation*}
\label{eq:Su_explicit}
S(u) = \int_0^T u_t\,dW_t = \frac{2}{\theta T}\int_0^{T/2} dW_t
     = \frac{2}{\theta T}\,W_{T/2}.
\end{equation*}

\subsection{Malliavin derivative of the call payoff.}
Using $S_T^\theta=e^{X_T^\theta}$ and the chain rule,
\begin{equation*}
\label{eq:DS}
D_t S_T^\theta = S_T^\theta\,D_t X_T^\theta = \theta S_T^\theta\,\1_{\{t\le T\}}=\theta S_T^\theta.
\end{equation*}
Then,
\begin{equation*}
\label{eq:Dell}
D_t\ell(X^\theta)
= e^{-rT}\,\1_{\{S_T^\theta>K\}}\,D_t S_T^\theta
= e^{-rT}\,\1_{\{S_T^\theta>K\}}\,\theta S_T^\theta.
\end{equation*}
Because $u$ is supported on $[0,T/2]$, the correction term in \eqref{eq:E1E2_general} becomes

\begin{align}
\label{eq:correction_simplify}
\int_0^T (D_t\ell(X^\theta))\,u_t\,dt
&= \int_0^{T/2} \Big(e^{-rT}\,\1_{\{S_T^\theta>K\}}\,\theta S_T^\theta\Big)\,\frac{2}{\theta T}\,dt \\
&= e^{-rT}\,\1_{\{S_T^\theta>K\}}\,S_T^\theta.
\nonumber
\end{align}

\paragraph{Final explicit $E_1^\theta,E_2^\theta$ for the Black--Scholes stressed call.}
With $g(X^\theta)=X_{T/2}^\theta-\log s$ and \eqref{eq:correction_simplify} we obtain,
\begin{align*}
\label{eq:E1_bs_final}
E_1^\theta
&=\E\!\left[\1_{\{X_{T/2}^\theta>\log s\}}
\Big(e^{-rT}(S_T^\theta-K)^+ \cdot \frac{2}{\theta T}W_{T/2}
      - e^{-rT}\,\1_{\{S_T^\theta>K\}}\,S_T^\theta\Big)\right],\\
E_2^\theta
&=\E\!\left[\1_{\{X_{T/2}^\theta>\log s\}}\cdot \frac{2}{\theta T}W_{T/2}\right].
\end{align*}
Thus the conditional call price under the counterfactual event $\{S_{T/2}=s\}$ is
$L(\theta)=E_1^\theta/E_2^\theta$.

\subsection{Conditional vega via weak derivatives.}
By the quotient rule, 
\begin{equation}
\label{eq:vega_ratio}
\nabla_\theta L(\theta)=\frac{E_2^\theta\,\nabla_\theta E_1^\theta - E_1^\theta\,\nabla_\theta E_2^\theta}{(E_2^\theta)^2}.
\end{equation}
To compute $\nabla_\theta E_1^\theta$ and $\nabla_\theta E_2^\theta$, apply the single-run
Hahn--Jordan estimator (Algorithm~\ref{alg:hjd}) to the path-functionals
\begin{align*}
\label{eq:C1C2_bs}
C_1(X^\theta) &:= \1_{\{g(X^\theta)>0\}}
\Big(\ell(X^\theta)S(u)-\int_0^T (D_t\ell(X^\theta))u_t\,dt\Big),\\
C_2(X^\theta) &:= \1_{\{g(X^\theta)>0\}}\,S(u),
\end{align*}
with $\ell,g,u$ given above and Euler discretization of \eqref{eq:bs_X}.

Figure~\ref{fig:mc_convg} illustrates the convergence of estimator \eqref{eq:L_bs_def} with respect to the number of Monte Carlo sample paths. It can be seen that we recover the classical $O(N^{-1/2})$ convergence rate, \textit{despite the measure-zero conditioning event}. Figure~\ref{fig:varscale} illustrates the estimator \eqref{eq:vega_ratio} variance when the $\theta$-gradients are computed with a score function (yielding $O(T)$ variance) and with the weak derivative method of Algorithm~\ref{alg:hjd} (yielding $O(1)$ variance). Thus, the Malliavin conditioning reformulation paired with the weak derivative gradient estimation forms a principled efficient pipeline for conditional gradient estimation. 

The code producing these figures is open-source, and can be found at \texttt{\small https://github.com/LukeSnow0/Malliavin-WD}.
\begin{figure}
\centering
\includegraphics[width=0.6\linewidth]{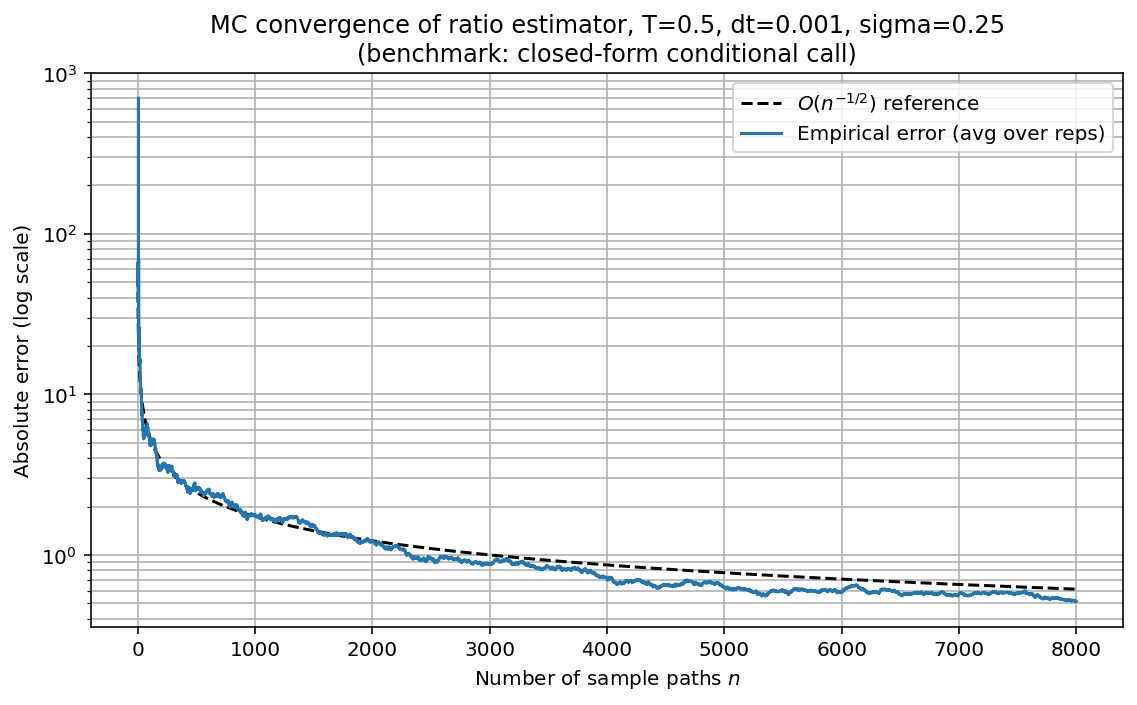}
\caption{Convergence of Monte-Carlo estimate of Malliavin quotient \eqref{eq:L_bs_def}, with respect to simulated paths $N$. We see that we recover a $O(N^{-1/2})$ convergence rate \textit{even though we condition on a measure-zero event}. }\label{fig:mc_convg}
\end{figure}

\begin{figure}[t]
    \centering
    \includegraphics[width=0.5\linewidth]{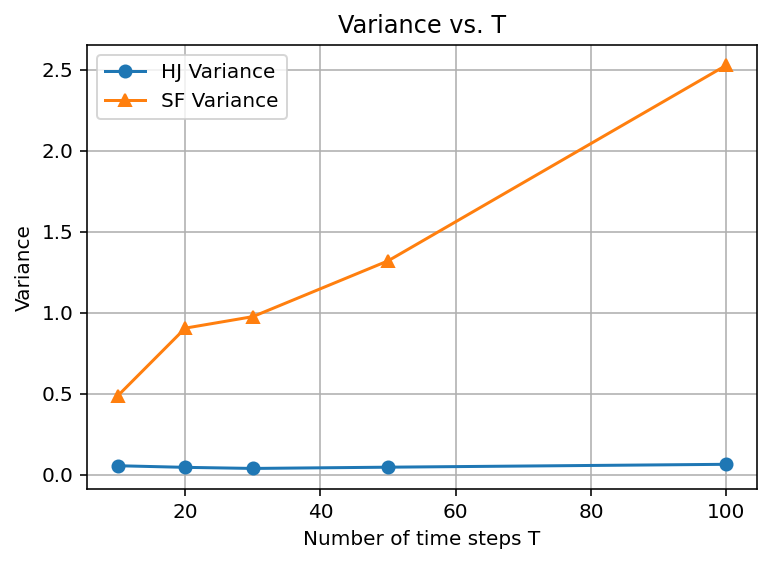}
    \caption{Variance scaling of the weak derivative estimator and the score function estimator, for varying time horizon $T$. We verify the stable $O(1)$ variance scaling of the weak derivative estimator, in contrast to the $O(T)$ variance scaling of the score function estimator.}\label{fig:varscale}
\end{figure}

\section{Conclusion}

We have presented a methodology for counterfactual estimation of conditional Greeks. Instead of relying on methods in kernel smoothing or direct conditional Monte-Carlo, we exploit a reformulation of the conditional expected gradient by Malliavin calculus. This allows for recovery of the $O(N^{-1/2})$ Monte-Carlo convergence rate even when conditioning on rare or measure-zero events, where direct Monte-Carlo becomes impossible and kernel smoothing methods infeasible and inefficient. Furthermore, we combine this approach with a weak-derivative gradient estimation algorithm which incurs stable $O(1)$ variance scaling in the time-horizon, in contrast to score function methods which scale as $O(T)$. The combination of these two methodologies in the framework of passive stochastic optimization is a novel approach which allows for efficient gradient evaluation, in particular conditional Greek evaluation, even under rare conditioning events or highly mis-specified functional evaluations. 

\bibliography{main}

\end{document}